\documentclass{article}

\usepackage{amsmath}

\usepackage{amssymb}
\usepackage{amsfonts}
\usepackage{amsthm}
\usepackage[margin=1in]{geometry}
\usepackage{hyperref}
\usepackage[width=14cm,format=plain,labelfont=sf,bf,textfont=sf,up,justification=justified,labelsep=period]{caption}
\usepackage{graphicx}
\usepackage{framed}
\usepackage{soul}
\usepackage{float}
\usepackage{sidecap}
\usepackage{xcolor}
\usepackage{tabularx}
\usepackage[mathscr]{euscript} 
\usepackage{amsbsy} 

\usepackage{cancel}

\hypersetup{
    colorlinks=true,       
    linkcolor=blue,          
    citecolor=blue,        
    filecolor=blue,      
    urlcolor=blue           
}

\makeatletter
\def\iddots{\mathinner{\mkern1mu\raise\p@
\vbox{\kern7\p@\hbox{.}}\mkern2mu
\raise4\p@\hbox{.}\mkern2mu\raise7\p@\hbox{.}\mkern1mu}}
\makeatother

\newcommand{\field}[1]{\mathbb{#1}}

\renewcommand{\Pr}{\mathbb{P}} 

\newcommand{\cE}{\mathcal E}

\newcommand{\NN}{\field{N}}
\newcommand{\RR}{\field{R}}

\def\ep{^\varepsilon}
\def\e{\varepsilon}

\DeclareMathOperator{\rank}{rank}
\DeclareMathOperator{\nnrank}{rank_+}
\DeclareMathOperator{\rankpsd}{\rank_{psd}}

\DeclareMathOperator{\supp}{supp}

\DeclareMathOperator{\trace}{Tr}

\DeclareMathOperator{\di}{dim}

\newcommand{\Da}{\Delta_{\alpha}}

\newcommand{\ov}[1]{\overline{#1}}

\renewcommand{\S}{{\bf S}}

\renewcommand{\tilde}[1]{\widetilde{#1}}

{\begin{small}\begin{list}{}%
         {\setlength{\leftmargin}{1cm}\setlength{\rightmargin}{1cm}}%
         \item[]%
}
{\end{list}\end{small}}

\theoremstyle{plain}
\newtheorem{prop}{Proposition}
\newtheorem{fact}{Fact}[section]
\newtheorem{lem}{Lemma}[section]
\newtheorem{thm}{Theorem}[section]
\newtheorem*{nnthm}{Factorization Theorem}

\newtheorem{cor}{Corollary}

\theoremstyle{definition}

\theoremstyle{remark}

\title{
\vspace{-1.1cm}
Exploring the bounds on the positive semidefinite rank}

\author{Andrii Riazanov\thanks{Skolkovo Institute of Science and Technology; Moscow Institute of Physics and Technology
(State University).}\\  andrii.riazanov@gmail.com\and Mikhail Vyalyi\thanks{Dorodnicyn Computing
    Centre, FRC CSC RAS; Moscow Institute of Physics and Technology (State University);
    National Research University Higher School of Economics. The study
    has been funded by the Russian Academic Excellence Project
    '5-100'.} \\
vyalyi@gmail.com
    }

\date{}

\begin{document}
\maketitle
\vspace{-0.5cm}
\begin{abstract}
The nonnegative and positive semidefinite (PSD-) ranks are closely connected to the nonnegative and positive semidefinite extension complexities of a polytope, which are the minimal dimensions of linear and SDP programs which represent this polytope. Though some exponential lower bounds on the nonnegative~\cite{Fiorini_2012} and PSD-~\cite{Lee_2015} ranks has recently been proved for the slack matrices of some particular polytopes, there are still no tight bounds for these quantities. We explore some existing bounds on the PSD-rank and prove that they cannot give exponential lower bounds on the extension complexity. Our approach consists in proving that the existing bounds are upper bounded by the polynomials of the regular rank of the matrix, which is equal to the dimension of the polytope (up to an additive constant). As one of the implications, we also retrieve an upper bound on the mutual information of an arbitrary matrix of a joint distribution, based on its regular rank.

\end{abstract}


\section{Introduction}
Linear optimization plays an important role in computer science and mathematics. Though there exist efficient algorithms of linear optimization over convex sets, for the polytopes with exponential number of facets they still work too long in general case. That is why one may want to represent such ``hard'' convex set as a projection (linear map) of some ``easier'' convex set, for example of some affine slice of the cone of nonnegative orthant or the cone of positive semidefinite matrices, since on slices of both these cones linear optimization has efficient algorithms. Such representations are called \emph{the nonnegative} and \emph{the positive semidefinite (PSD-)} \emph{extensions}, respectively. 

Since many problems of combinatorial optimization can be represented as linear programs over a polytope, studying the extensions of convex polytopes is an important and challenging problem. The natural question is to find the minimal dimension for which there exists an extension of the given polytope. It can be also formulated as determining the smallest dimensions of LP or SDP programs which represent optimization over the given polytope, and such sizes are called \emph{the nonnegative} and \emph{the semidefinite extension complexities}, respectively. 
%
%

In the context of P $\not=$ NP we do not expect to find small nonnegative or PSD- extension complexities for NP-hard problems, since that would mean that there exist polynomial algorithms for solving these problems. However, there is still no general approach for proving the lower bounds on these quantities, and only a few exponential lower bounds for some particular problems has recently been proved. All such results use the connection between extension complexity and matrix factorizations, which was first discovered in~\cite{Yannakakis_1991} for the nonnegative extension complexity and nonnegative matrix factorizations. Further, this approach was extended in~\cite{Gouveia_2013} for the general case of cone factorizations, and the same result for PSD-factorizations was also obtained in~\cite{Fiorini_2012}. This instrument gave an opportunity to explore the nonnegative and PSD- extension complexities of polytopes via studying some characteristics of their slack matrices called \emph{the nonnegative} and \emph{the PSD-} ranks. For example, in the 1980s there were attempts to prove P = NP by providing the polynomial-sized linear program to solve the NP-hard travelling salesman problem (TSP). However, using the described approach, Yannakakis proved in~\cite{Yannakakis_1991} that any \emph{symmetric} LP which solves TSP has exponential size, which meant invalidity of all such attempts, since all the presented LPs were symmetric. The extension of this result for \emph{any} (not only symmetric) TSP was first presented in~\cite{Fiorini_2012}, where the authors used the connection between the nonnegative rank of the matrix and the nondeterministic communication complexity of its support. In this work, the exponential lower bounds on the nonnegative rank were also proved for CUT and Stable Set polytopes. The first analogical bounds for the PSD-extension complexity were presented in \cite{Lee_2015} using the sum-of-squares SDP hierarchy.

Since exponential lower bounds were obtained for some particular cases only, it is still a challenging problem to obtain reasonable estimations and bounds for the nonnegative and PSD- ranks. This problem is widely discussed during the last decade. 
For instance, exponential bounds on the nonnegative rank, and thus on the nonnegative extension complexity, were proved in~\cite{Rothvoss_2014} for the matching polytope , where the author used the extension of Razborov's result~\cite{Razborov}. We address the reader to the review~\cite{Fawzi_2015} for more details about recent research on the PSD-rank. 

There is also a problem of determining the computational complexity of computing the nonnegative and PSD- ranks. Both problems are known to be NP-hard, and recent research~\cite{Shitov_2016} shows that the problem of computing the PSD-rank is complete in $\exists\RR$ -- the 
existential theory of the reals.


\subsection*{Contribution}
In this paper we explore the lower bounds on the PSD-rank introduced in~\cite{Lee_2016}, which we will further address as \emph{bounding functionals (of a matrix)}. 
We show that these functionals cannot give exponential bounds on the PSD-rank, and thus on the positive semidefinite extension complexity. Our approach consists in proving that the bounding functionals of the slack matrix are bounded above 
by the polynomial of the regular rank of this matrix and the logarithm of the matrix size. 
Since for any polytope $P$ we have $\rank S_P = \di(P) + 1$, it would mean that the bounds are polynomial in the dimension of the polytope. 

As one of the implications of our approach, we achieve the upper bound on the mutual information for an arbitrary matrix of a joint distribution. More precisely, we show that the mutual information is bounded above by the logarithm of the rank of the matrix.

\subsection*{Outline of the paper}
This paper is organized as follows. In Sect. \ref{prel} we introduce all the necessary notations and explain some connections between the PSD-rank and the quantum communication complexity. In Sect. \ref{known} we present the bounding functionals from \cite{Lee_2016} and explain how the lower bound on the PSD-rank can be obtained via the mutual information. Finally, in Sect. \ref{upper} the upper bounds on the bounding functionals are proved. In particular, Theorem \ref{mutual_thm} shows that the mutual information of two discrete random variables is bounded above by the logarithm of the regular rank of the matrix of their joint distribution.

\section{Preliminaries}
\label{prel}

\subsection{Nonnegative and PSD- matrix factorizations}

\emph{The nonnegative matrix factorization} of the nonnegative matrix $A \in \RR^{m\times n}$ is the decomposition $A = BC$, where $B \in \RR^{m\times k}$, $C \in \RR^{k\times n}$, and $B, C$ are nonnegative matrices. Alternatively, such factorization can be thought of as two sets of vectors $\{b_i\}_{i=1}^m,\ \{c_j\}_{j=1}^n,\ b_i, c_j \in \RR^{k}_+$, such that $A(i, j) = \langle b_i, c_j \rangle$. Then \emph{the nonnegative rank} of $A$, denoted $\nnrank A$, is the smallest $k \in \NN$ for which such nonnegative factorization of $A$ exists.

Similarly, \emph{the positive semidefinite rank} $\rankpsd A$ is the minimal integer $r$ for which there exist two sets of complex Hermitian positive semidefinite matrices $\{B_i\}_{i=1}^m,\ \{C_j\}_{j=1}^n,\ B_i, C_j \in \S^{r}_+$, such that $A(i, j) = \langle B_i, C_j \rangle = \trace (B_iC_j)$. Such factorization is called \emph{the positive semidefinite factorization}, and it has many applications in combinatorial optimization and communication complexity. If to restrict the matrices in the factorization to be real symmetric positive semidefinite, one will obtain the definition of \emph{the real PSD-rank} $\rank^{\RR}_{psd}$. It can be shown (\cite{Lee_2016}), that the restriction for matrices to be real can increase $\rankpsd$ at most by the factor of 2, e.g. $\rankpsd \leq \rank^{\RR}_{psd} \leq 2\rankpsd$. Since in our context we only study asymptotic bounds on the ranks, there is no difference between considering $\rankpsd$ or $\rank^{\RR}_{psd}$.

We would like to emphasize that rescaling the nonnegative matrix by multiplying its rows or columns by any positive factors does not change its nonnegative and PSD- ranks. 
Indeed, multiplication of the $i^{th}$ row of $A$ by $\alpha$ corresponds to the multiplication of $b_i$ by the same factor $\alpha$ in the nonnegative factorization. Similarly, it corresponds to the multiplication of $B_i$ by $\alpha$ in the PSD-factorization. Obviously, the situation with the columns of $A$ is the same.
%
%


\subsection{Extension complexity}
\label{extension}

\emph{The nonnegative extension complexity} of the polytope $P$ is the smallest number $d$ such that $P$ can be expressed as a projection of an affine slice of the nonnegative $d$-dimensional orthant $\RR_+^d$. Similarly, \emph{the semidefinite (PSD-) extension complexity} of $P$ is the minimum number $r$ for which there exists an affine slice of the cone of complex Hermitian $r\times r$ positive semidefinite matrices $\S^r_+$ that projects onto $P$.

In other words, for optimizing over some polytope $P \in \RR^d$ one may want to represent is as $P = \pi(K~\cap~L)$, where $K \subseteq \RR^n$ is some close convex cone, $L$ is some affine subspace of $\RR^n$, and $\pi$ is a linear map (projection). Such representations are called \emph{$K$-lifts}, (\cite{Gouveia_2013}), or \emph{$K$-extensions}. If to choose $K$ from the families of the cones of nonnegative orthants $\RR^k_+$ or positive semidefinite matrices $\S^r_+$, the nonnegative and PSD- extension complexities for the given polytope correspond to minimal $k$ and $r$ for which such representations exist.


\subsection{Factorization theorem}
\label{factorization}

As it was discussed in Introduction, \cite{Yannakakis_1991}, \cite{Gouveia_2013}, and \cite{Fiorini_2012} proved that the extension complexities and matrix factorizations are interconnected. Here we present the Factorization theorem, which explains the relations between these two notions.


Let $P$ be a polytope in $\RR^d$ with $n$ vertices and $m$ facets, thus $P = \{ x \in \RR^d\ |\ \langle x, a_j \rangle \leq b_j,\ j \in \overline{1, m}\}$. Then \emph{the slack matrix of the polytope $P$} is defined as the nonnegative matrix $S_P \in \RR^{n\times m}$ with $S_P(i, j) = b_j - \langle v_i, a_j \rangle$, where $v_i$ is the $i^{th}$ vertex of $P$. Then the Factorization theorem can be formulated as follows:
\begin{nnthm}
 The nonnegative extension complexity of $P$ is equal to $\nnrank S_P$. Similarly, the PSD-extension complexity of $P$ is equal to $\rankpsd S_P$.
 \end{nnthm}
%
%

This approach allows applying techniques for estimating or bounding such algebraic notions as sizes of matrix factorizations to answer geometrical questions about the complexities of the polytopes.

\subsection{Quantum communication complexity}
In this section, we describe the connection between the quantum communication complexity and $\rankpsd$. First, we will consider \emph{one-way quantum communication protocol}.

A quantum state $\rho$ is a positive semidefinite matrix with $\trace{\rho} = 1$. 
A measurement $\cE$ is the set of positive semidefinite matrices $\{E_i\}_{i \in \Omega}$, indexed by the finite set of nonnegative real numbers $\Omega$, with the condition $\Sigma_{i \in \Omega}E_i = I$. The measurements are also called POVM (``Positive Operator Value Measure'') in the literature. POVMs work in the following way: when we apply the measurement $\cE$ to the state $\rho$, the outcome is $i$ with probability $\trace(E_i\rho)$.

Then the process of communication is set as follows: initially, Alice has the integer $x$, and Bob has $y$. Then Alice sends an $r\times r$-dimensional quantum state $\rho_x$ to Bob, who measures it with POVM $\cE_y$ and outputs the result. We say that such a protocol computes the nonnegative matrix $M$ in expectation, if the expected value of Bob's output on the input $(x, y)$ is equal to $M(x,y)$ (the entry of the matrix $M$ in $x^{th}$ row and $y^{th}$ column). Then \emph{the quantum communication complexity} of the matrix $M$ is the logarithm of such a minimal size of dimension $r$, for which there exists a one-way quantum protocol which computes $M$ in expectation.

Fiorini et. al. \cite{Fiorini_2012} and Jain et. al. \cite{Jain_2013} proved that the minimal amount of quantum information needed for Alice and Bob to generate the nonnegative matrix $M$ is completely determined by the PSD-rank of this matrix. More precisely, they showed that the quantum communication complexity of $M$ is equal to $\lceil \log \rankpsd M \rceil$.

\section{Bounding functionals on the PSD-rank}
\label{known}
In this section, we present some existing general lower bounds on $\rankpsd$ from \cite{Lee_2016}, which we address as \emph{bounding functionals}. Except for the bound via mutual information, the bounding functionals are introduced here without justification. We address the reader to the original article for more details on the bounds. For convenience, we preserve the notations for the bounding functionals from the original article. 

\subsection{Bound via Mutual Information}
If $X$ and $Y$ are two random variables, then \emph{the mutual information} is defined as follows: 
\[I(A:B) = H(A) + H(B) - H(A, B) = H(A) - H(A | B) = H(B) - H(B|A),\]
where $H$ is Shannon entropy. The mutual information can be interpreted as the number of bits of information about $A$ that are revealed by the value of $B$. We will now use Holevo's theorem \cite{Watrous_11} to bound the mutual information. It claims that the number of classical bits of information that Alice can communicate to Bob by sending $n$ qubits does not exceed $n$. From the previous passage we know that we need exactly $\lceil\log \rankpsd M\rceil$ qubits of information to compute the matrix $M$. Normalizing $M$ and considering it as a matrix of joint distribution $\Pr(A, B)$, we then have:
\begin{fact}
\label{F2} Let $M$ be a matrix of a joint distribution of two discrete random variables $A, B$ with finite support, $M(a, b) = \Pr[B = b, A = a]$. Then    
\begin{flalign*}
&\hspace{5cm}  \rankpsd M \geq B_2(P) = 2^{I(A:B)}. &
\end{flalign*}
\end{fact}

\subsection{Bounding functionals from \cite{Lee_2016}}

For two probability distributions $p = \{ p_i\}_{i=1}^n$ and $q = \{ q_i\}_{i=1}^n$ fidelity is defined as $F(p, q) = \Sigma_{i=1}^n \sqrt{p_iq_i}$.

Recall that \emph{the left stochastic matrix} is the matrix with nonnegative entries, with each column summing~to~$1$. Further in the text we will omit ``left'' and just use the term ``stochastic matrix'' instead. 

Then we have the following lower bounds:

\begin{fact}
\label{F3} Let $M \in \RR^{n\times m}$ be a stochastic matrix. Then
\begin{flalign*}
&\hspace{5cm} \rankpsd M \geq B_3(M) = \max_{\{q_i\}_{i=1}^m} \dfrac1{\sum_{i,j = 1}^{m}q_iq_jF(M_i, M_j)^2}&
\end{flalign*}
where the $\max$ is taken over all probability distributions $q = \{q_i\}_{i=1}^m$, and $M_i$ is the $i^{th}$ column of $M$.
\end{fact}

\begin{fact}
\label{F4} Let $M \in \RR^{n\times m}$ be a stochastic matrix. Then
\begin{flalign*}
&\hspace{5cm}  \rankpsd M \geq B_4(M) = \sum_{i=1}^n\max_j M(i, j).&
\end{flalign*}
\end{fact}

\begin{fact}
\label{F5} Let $M \in \RR^{n\times m}$ be a stochastic matrix. Then
\begin{flalign*}
&\hspace{5cm} \rankpsd M \geq B_5(M) = \sum_{i=1}^n \max_{\{q^{(i)}_j\}_{j=1}^m} \dfrac{\sum_{k=1}^mq_k^{(i)}M(i, k)  }{\sqrt{\sum_{s,t = 1}^{m}q^{(i)}_sq^{(i)}_tF(M_s, M_t)^2}}&
\end{flalign*}
where the $\max$ is taken over all probability distributions $q^{(i)} = \{q^{(i)}_j\}_{j=1}^m$, and $M_i$ is the $i^{th}$ column of $M$.
\end{fact}

\section{Upper bounds on the bounding functionals}
\label{upper}
All the bounds from section \ref{known} were explored and compared in \cite{Lee_2016}. It turned out that in different cases $B_2, B_3, B_4$, or $B_5$ can give better bounds on $\rankpsd$ than others, and some of them can be tight in some particular cases. However, the key question of whether these functions can give exponential lower bounds on the PSD-rank with respect to the regular rank was not addressed. In this section we answer this question negatively.
%
%

In the context of combinatorial optimization, we would like to show that for the polytope of some \mbox{NP-hard} problem the semidefinite extension complexity is exponential in the dimension. Following the arguments from Section~\ref{extension}, it suffices to show that the PSD-rank of the corresponding slack matrix is exponential. It is easy to show (\cite{Gouveia_2013_2}) that the regular rank of the slack matrix equals to the dimension of the polytope plus one: $\rank S_P = \di P + 1$.
For all the presented bounding functionals we provide the upper bounds polynomial in the regular rank of the matrix and the logarithm of the matrix size, which means that they cannot be exponential in the dimension.

\subsection{Row elimination transformation}
\label{elimination}
We will now describe the row elimination transformation, which will be used for proving the required bounds.
%
%

Let $M \in \RR^{n\times m}$ be a nonnegative matrix with $\rank M = r < n$. Without loss of generality, assume that first $r+1$ rows $\ov{m_1}, \ov{m_2} \dots, \ov{m_{r+1}}$ are non-zero. 
They are linearly dependent, so there exists a nontrivial set of real numbers $\{\alpha_i\}_{i=1}^{r+1}$, such that $\sum_{i=1}^{r+1}\alpha_i\ov{m_i} = \ov{0}$. Since all entries of $M$ are nonnegative, there are both negative and positive numbers among $\{\alpha_i\}_{i=1}^{r+1}$. For such a set of real numbers $\{\alpha_i\}$ we denote by $\Delta_{\alpha}$ the closed interval $\Delta_{\alpha} = \left[ -\dfrac1{\max_i \alpha_i},\ -\dfrac1{\min_i \alpha_i} \right]$, which is properly defined due to the last remark.

Then we define the matrix $M\ep$ as follows: for $1 \leq i \leq (r+1)$ the $i$-th row of $M\ep$ equals $\ov{m_i}(1 + \e\alpha_i)$, for~$i > (r+1)$ the $i$-th row of $M\ep$ coincides with the $i$-th row of $M$. We call the matrix $M\ep$ \emph{$\e$-transformation} of $M$.

First of all, note that $\quad (1 + \e\alpha_i)\ \geq 0\quad\forall i \in \ov{1, (r+1)} \quad \Leftrightarrow \quad \e \in \Delta_{\alpha}$. Moreover, it holds that when $\e$ is equal to one of the ends of $\Delta_{\alpha}$, at least one of the coefficients $(1 + \e\alpha_i)$ is equal to zero. It means that for $\e \in \Delta_{\alpha}$ the matrix $M\ep$ is nonnegative matrix, and when $\e$ is either the left or the right end of $\Da$, $M\ep$ has more zero rows than $M$.

Next, we prove that sums of columns do not change after row elimination transformation. Indeed, 
\[ \sum_{i=1}^n m\ep_{ij} = \sum_{i=1}^{r+1} m\ep_{ij} + \sum_{i=r+2}^n m_{ij} = \sum_{i=1}^{r+1} m_{ij}(1 + \e\alpha_i) + \sum_{i=r+2}^n m_{ij} = \sum_{i=1}^{n} m_{ij} + \e\underbrace{\sum_{i=1}^{r+1}\alpha_im_{ij}}_{0} = \sum_{i=1}^{n} m_{ij}. \]
In particular, it means that if $M$ is stochastic, then for $\e \in \Da\ M\ep$ is also stochastic. Similarly, if $M$ is a matrix of a joint distribution, then $M\ep$ is also a matrix of some joint distribution for $\e$ from $\Da$.

\subsection{Upper bound on $B_2$ (Mutual Information)}

Let $M\in \mathbb{R}^{n\times m}$ be the matrix of a joint distribution of two discrete random variables $X, Y:$ \[m_{ij} = \mathbb{P}\left[X = x_i, Y = y_j\right] \geq 0, \qquad \sum_{i=1, j=1}^{n, m} m_{ij} = 1.\]
Let $p_i,\ i \in \overline{1, n}$, and $q_j,\ j \in \overline{1, m}$, be the marginal probabilities of $X$ and $Y$ respectively:
\[  p_i = \mathbb{P}\left[X = x_i\right] = \sum_{j=1}^m m_{ij},\ i \in \overline{1, n};   \qquad q_j = \mathbb{P}\left[Y = y_j\right] = \sum_{i=1}^n m_{ij},\ j \in \overline{1, m}.    \]
Then the mutual information between $X$ and $Y$ can also be defined as: 
\[I(X : Y) = D_{KL} \left( p(X, Y)\, ||\, p(X)p(Y)\right) =
\sum_{i=1}^n\sum_{j=1}^m p(x_i, y_j) \log_2\left( \dfrac{p(x_i, y_j)}{p(x_i)p(y_j)}\right) =\sum_{i=1}^n\sum_{j=1}^m m_{ij} \log_2\left( \dfrac{m_{ij}}{p_iq_j}\right),\] 
where we set $0\log\dfrac{0}{q} = 0$ (the logarithm here and further is to the base 2). We also denote $I(M) = I(X:Y)$.

\begin{thm}
\label{mutual_thm}
Let $M \in \RR^{n\times m}$ be the matrix of a joint distribution of $X$ and $Y$. Then
 \[B_2(M) = 2^{I(X:Y)} \leq \rank M .\]
\end{thm}
\begin{proof}
Denote $r = \rank M$. 
We will now transform the original matrix $M$ in such a way, that the mutual information will not decrease, but the new matrix $\tilde{M}$ will have at most $r$ non-zero rows. 

Suppose $M$ has more than $r$ non-zero rows. Then we apply the row elimination transformtaion and consider the $\e$-transformation $M\ep$ of the original matrix. Since we have already shown that it is also a matrix of some joint distribution, we explore how the mutual information changes after such transformations.
%
%

First, since the $\e$-transformation does not change the sums in the columns of $M$, we have $q\ep_j = q_j$. Then, since $p\ep_i$ is the sum of entries in the $i$-th row, we obtain $p\ep_i = p_i(1+\e\alpha_i)$.

 Note that since $M$ and $M\ep$ coincide on rows with indexes larger than $r+1$, we may omit the summation over these rows:
\[ I(M\ep) - I(M) =  \sum_{i=1}^{r+1}\sum_{j=1}^s\left[ m\ep_{ij} \log\left( \dfrac{m\ep_{ij}}{p\ep_iq\ep_j}\right) - m_{ij}\log\left( \dfrac{m_{ij}}{p_iq_j}\right)  \right] = \]

\[ = \sum_{i=1}^{r+1}\sum_{j=1}^s\left[ m_{ij}(1 + \e\alpha_i) \log\left( \dfrac{m_{ij}\cancel{(1 + \e\alpha_i)}}{p_i\cancel{(1 + \e\alpha_i)}q_j}\right) - m_{ij}\log\left( \dfrac{m_{ij}}{p_iq_j}\right)  \right] = \]

\[ = \sum_{i=1}^{r+1}\sum_{j=1}^s\left[ \e \alpha_im_{ij}\log\left( \dfrac{m_{ij}}{p_iq_j}\right)  \right] = \e\cdot\Lambda.   \]


Now recall that the $\e$-transformation is valid for $\e \in \Da$, where the left end of $\Da$ is negative, and the right end is positive. It means that 
we can choose an end of the interval of $\Da$ such that $I(M\ep) \geq I(M)$. It only remains to note that with the chosen value of $\e$ at least one of the first $(r+1)$ rows in $M\ep$ becomes zero.

To get an upper bound on the mutual information, we apply $\e$-transformations with such  suitable $\e$'s that the number of non-zero rows strictly decreases and the mutual information does not decrease. At the end of such procedure we obtain the matrix $\tilde{M}$ with at most $r$ non-zero rows for which $I(M) \leq I(\tilde{M})$.
Since $\tilde{M}$ is the matrix of joint distribution, we have $I(\tilde{M}) = I(\tilde{X} : \tilde{Y})$, where the support of $\tilde{X}$ has cardinality at most $r$. Using the equality $I(\tilde{X} : \tilde{Y}) = H(\tilde{X}) - H(\tilde{X} | \tilde{Y})$ and the non-negativity of the conditional entropy, we finally have:
\[ I(M) \leq I(\tilde{M}) = I(\tilde{X} : \tilde{Y}) \leq  H(\tilde{X}) \leq \log |\supp (\tilde{X})| \leq \log r .\]


\end{proof}

\subsection{Upper bound on $B_3$}
\label{sectB3}


We will show that $B_3(M)$ is upper bounded by $poly(\mathrm{rank}(M), \ln m)$: 
\begin{thm}
\label{TB3}
Let $M \in \RR^{n\times m}$ be a stochastic matrix, $\rank M = r$. Then
\[  B_3(M) \leq (\ln m +1)^2r^2. \]
\end{thm}

We start with proving the following well-known fact:

\begin{lem}
\label{fp}
For distributions $p, q$ it holds $F(p, q) \geq 1 - \dfrac{|p - q|}{2}$, where $|p-q|$ is $l_1-$norm of the vector $(p-q)$, and thus $\dfrac{|p - q|}{2}$ is the statistical distance between the distributions.
\begin{proof}
\[1 - \sum_{k=1}^m\sqrt{p_kq_k} = \dfrac1{2}\left( \sum p_k + \sum q_k - 2 \sum\sqrt{p_kq_k} \right) =  
 \dfrac1{2}\sum\left| \sqrt{p_k} - \sqrt{q_k} \right|^2 \leq    \dfrac1{2}\sum\left| p_k - q_k \right|     \]
\[ \Rightarrow F(p, q) = \sum_{k=1}^m\sqrt{p_kq_k} \geq 1 - \dfrac1{2}\sum\left| p_k - q_k \right|  = 1 - \dfrac{|p - q|}{2}.\]
\end{proof}
\end{lem}\par\bigskip

Now, we have \begin{equation}
\label{B3def}
B_3(M) = \underset{\{q_i\}_{i=1}^m}{\max} \dfrac{1}{\sum_{i,j}q_iq_jF(M_i, M_j)^2} =  \dfrac{1}{\underset{\{q_i\}_{i=1}^m}{\min}\sum_{i,j}q_iq_jF(M_i, M_j)^2}.
\end{equation}Then we need to prove the lower bound on $\underset{q\in\Delta_m}{\min}\sum_{i,j}q_iq_jF(M_i, M_j)^2$. 
%
%


We will find the lower bound on this quadratic form for an arbitrary distribution $q$. Without loss of generality, assume $q_1 \geq q_2 \geq \dots \geq q_n$. 
\begin{lem}
There exists $s \in \overline{1,m}$ such that $sq_s \geq \frac1{\ln m + 1}$.
\begin{proof}
Suppose the opposite: $sq_s \leq \frac1{\ln m + 1} \ \forall s \in  \overline{1,m}$. Then 
\[1 = q_1 + q_2 + \cdots + q_m \leq   \frac1{\ln m + 1} + \frac1{2\left(\ln m + 1\right)} + \cdots + \frac1{m\left(\ln m + 1\right)} =\]
\[= \frac1{\ln m + 1}\left(1 + \dfrac1{2} + \dfrac1{3} + \cdots + \dfrac1{m}\right) <  \frac1{\ln m + 1}\left(1 + \int_{1}^m \dfrac1{x}dx\right)    = 1.        \]
\end{proof}
\end{lem}

Then we have \begin{align} \sum\limits_{i,j=1}^{m}q_iq_jF(M_i, M_j)^2 \geq \sum\limits_{i, j =1}^{s}q_iq_jF(M_i, M_j)^2 \geq \sum\limits_{i, j =1}^{s}q_s^2F(M_i, M_j)^2 = \notag\\ = s^2q_s^2 \cdot\dfrac{\sum\limits_{i, j =1}^{s}F(M_i, M_j)^2} {s^2} \geq \dfrac1{(\ln m +1)^2}\cdot \dfrac{\sum\limits_{i, j =1}^{s}F(M_i, M_j)^2} {s^2} \label{first}\end{align}
Now, using the RMS-AM inequality and Lemma \ref{fp}, we get:

\begin{equation}   \dfrac{\sum\limits_{i, j =1}^{s}F(M_i, M_j)^2} {s^2} \geq \left( \dfrac{\sum\limits_{i, j =1}^{s}F(M_i, M_j)} {s^2} \right)^2        \geq \left(  \dfrac{\sum\limits_{i,j=1}^{s}\left(1-\dfrac{|M_i-M_j|}{2}\right)}{s^2}  \right)^2  = \left( 1 - \dfrac{\dfrac1{2}\sum\limits_{i,j=1}^{s}|M_i - M_j|}{s^2} \right)^2  \label{ko}\end{equation}

For any stochastic matrix $M \in \mathbb{R}^{n\times m}$ denote $S(M) = \dfrac{\dfrac1{2}\sum\limits_{i,j=1}^{m}|M_i - M_j|}{m^2}$ -- the arithmetic mean of statistical distances between $m$ columns of $M$. It now suffices to show the upper bound on $S(M)$.

\begin{lem}
\label{premain}
Let $M \in \mathbb{R}^{n\times m}$ be a stochastic matrix with $\mathrm{rank}(M) = r$. Then there exists a stochastic matrix $\tilde{M} \in \mathbb{R}^{r\times m}$ such that $S(M) \leq S(\tilde{M})$.
\begin{proof}

 We apply the row elimination algorithm. Suppose $M$ has more then $r$ non-zero rows. Consider then the $\e$-transformation $M\ep$ of the original matrix. Since the $\e$-transformation does not change the sums of entries in every column of the matrix, $M\ep$
is also stochastic. We now explore how $S(M)$ changes after the $\e$-transformation:

\begin{align*}
\notag
S(M\ep) - S(M) &= \dfrac{1}{2m^2}\left( \sum\limits_{i,j=1}^{m}\left(|M\ep_i - M\ep_j| - |M_i - M_j|\right)      \right)   = \\ 
 &=   \dfrac{1}{2m^2}\left(\sum\limits_{k=1}^n\left[ \sum\limits_{i,j=1}^{m}\left(|m\ep_{ki} - m\ep_{kj}| - |m_{ki} - m_{kj}|\right)    \right]  \right)   =           \\
  &= \dfrac{1}{2m^2}\left(\sum\limits_{k=1}^{r+1}\left[ \sum\limits_{i,j=1}^{m}\left(|m_{ki} - m_{kj}|(1+\varepsilon\alpha_k) - |m_{ki} - m_{kj}|\right)    \right]  \right)      =       \\
 &=  \dfrac{1}{2m^2}\left(\sum\limits_{k=1}^{r+1}\left[ \sum\limits_{i,j=1}^{m}|m_{ki} - m_{kj}|\varepsilon\alpha_k   \right]  \right)    = \varepsilon   \cdot \Lambda.          
\end{align*}

So, the difference $S(M\ep) - S(M)$ is linear in terms of $\varepsilon$. Remind again that the $\e$-transformation is valid for $\e \in \Da$, where the left end of $\Da$ is negative, and the right end is positive. It means that we can choose an end of the interval of $\Da$ such that 
$S(M\ep) \geq S(M)$ and with the chosen value of $\e$ at least one of the first $(r+1)$ rows in $M\ep$ becomes zero.
When we apply such $\e$-transformations with suitable $\e$'s, 
the number of non-zero rows strictly decreases, and $S(M)$ does not decrease. At the end of such procedure we will obtain the matrix $\tilde{M}$ with at most $r$ non-zero rows for which $S(M) \leq S(\tilde{M})$.
\end{proof}

\end{lem}

\begin{lem}
\label{main}
Let $M \in \mathbb{R}^{r\times m}$ be a stochastic matrix. Then
%
%

\[S(M)  \leq 1 - \dfrac1{r}.\]
\begin{proof}
If $m \leq r$, then $\dfrac{\dfrac1{2}\sum\limits_{i,j=1}^{m}|M_i - M_j|}{m^2} \leq \dfrac{m^2-m}{m^2} = 1 - \dfrac1{m} \leq 1 - \dfrac1{r}$, where we just used $|M_i - M_j| \leq 2$.\\

\vspace{2pt}
\noindent Now suppose $ m > r $. Denote $Z(M) = \dfrac1{2}\sum\limits_{i,j=1}^{m}|M_i - M_j| = \dfrac1{2}\sum\limits_{k=1}^r\sum\limits_{i,j=1}^{m}|m_{ki} - m_{kj}|$.\\

\vspace{1pt}
\noindent We now construct the matrix $B$ by sorting every row of $M$. Obviously, $Z(M) = Z(B)$, since it is just a permutation of terms. Then
\[ Z(M) =   Z(B) =     \dfrac1{2}\sum\limits_{k=1}^r\sum\limits_{i,j=1}^{m}|b_{ki} - b_{kj}|     =  \sum\limits_{k=1}^r\sum\limits_{i=1}^m\sum\limits_{j=i}^m(b_{ki} - b_{kj}).        \]
For each $b_{ki}$ in this sum it occurs $(m-i)$ times with the sign $(+1)$ and $(i-1)$ times with the sign $(-1)$. Hence, 
%
%

\begin{align}
  Z(M) = \sum\limits_{k=1}^r \left( (m-1)b_{k1} + (m-3)b_{k2} + \cdots - (m-3)b_{k(m-1)} - (m-1)b_{km}                     \right)  = \nonumber \\
 = (m-1)\sum_{k=1}^rb_{k1} +    (m-3)\sum_{k=1}^rb_{k2} + \cdots - (m-3)\sum_{k=1}^rb_{k(m-1)} - (m-1)\sum_{k=1}^rb_{km}      
\label{wow}
\end{align}

Clearly, $Z(M)$ takes its maximal value when the sum in the first columns of $B$ is maximal. 
Since $b_{ki} \leq 1$ and the sums of all the entries in $B$ and $M$ coincide and are equal to $m$, to maximize $Z(M)$ we need to have $m$ ones in total in the first columns of $B$. Denote $m = sr + p,\ p < r$. If $r=1$, then the matrix $M$ consists of ones only (since it is stochastic), then $S(M) = 0$ and the inequality in the lemma is obvious. If $r > 1$, then it is easy to show that $(s+1) \leq \lceil\frac{m}{2}\rceil$. Note that exactly first $\lceil\frac{m}{2}\rceil$ summands are nonnegative in \eqref{wow}, so to maximize $Z(M)$ first $(s+1)$ columns of $B$ should be filled with ones:
%
%

\begin{figure}[H]
  \begin{minipage}{.5\textwidth}
    \begin{equation*}
\qquad B^* = 
\begin{pmatrix} 
  1   & 1 & \cdots &1 &1    & 0 & 0 &\cdots & 0\\ 
  \vdots & \vdots & \cdots &\vdots &\vdots & \vdots & \vdots &\cdots & 0\\

  \vdots & \vdots & \cdots &1  &1& 0 & 0 &\cdots & 0\\
  \vdots & \vdots & \cdots &1  &0& 0 & 0 &\cdots & 0\\
  \vdots & \vdots & \cdots &\vdots &\vdots & \vdots & \vdots &\cdots & 0\\

 1   & 1 & \cdots &1 &0 & 0 & 0 &\cdots & 0\\ 

\end{pmatrix}
\end{equation*}
  \end{minipage}%
  \begin{minipage}{.5\textwidth}
    \centering
    \includegraphics[width=6cm]{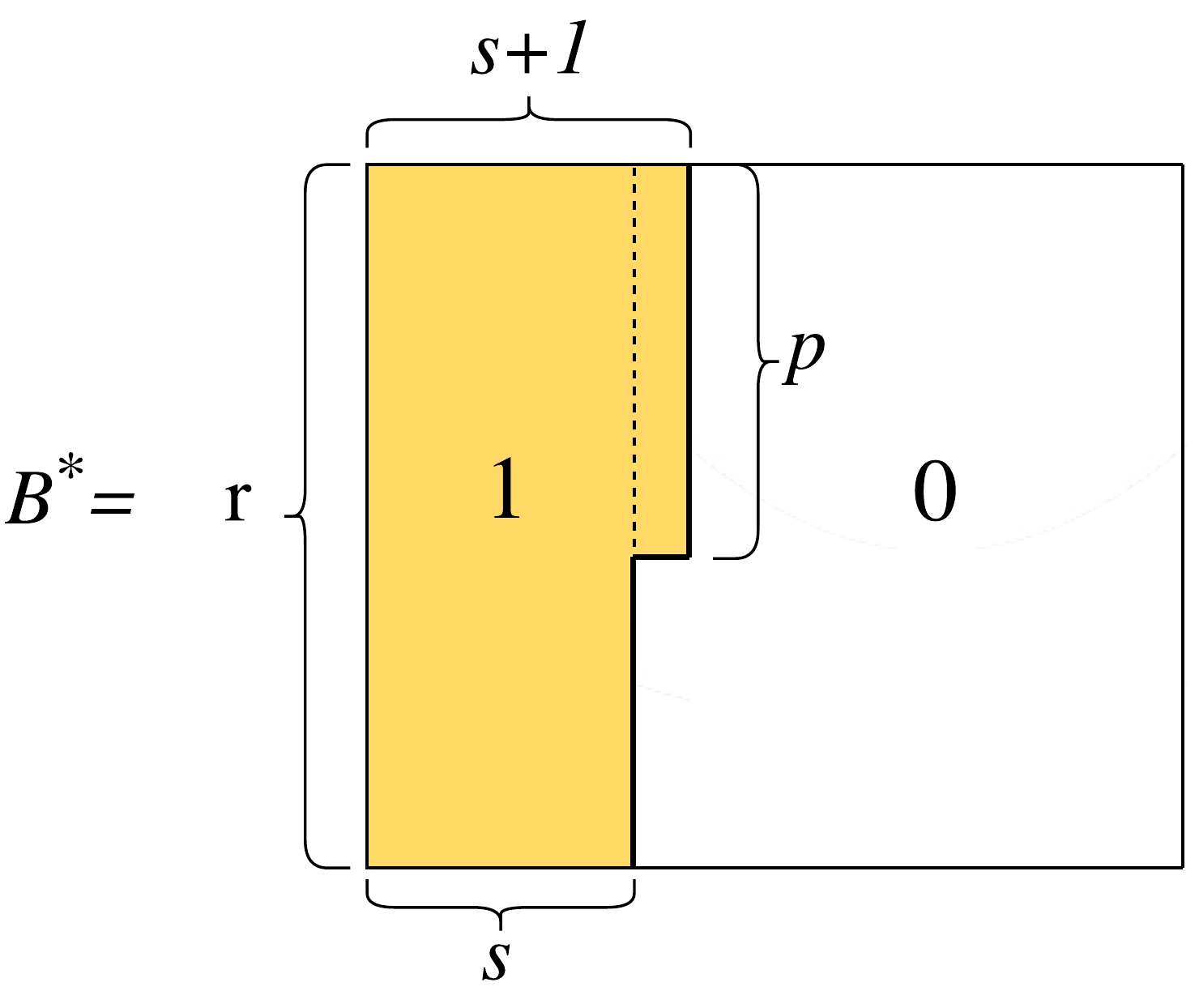}
  \end{minipage}
\end{figure}





Such matrix $B^*$ would correspond to the following matrix $M^*$:
\begin{equation*}
M^* = 
\begin{pmatrix} 
1 &   &        &        &   & 1 &   &        &        &   &        & 1 &        &   \\
  & 1 &        &        &   &   & 1 &        &        &   & \cdots &   & \ddots &   \\
  &   & \ddots &        &   &   &   & \ddots &        &   & \cdots &   &        & 1 \\
  &   &        & \ddots &   &   &   &        & \ddots &   & \cdots &   &        & 0 \\
  &   &        &        & 1 &   &   &        &        & 1 &        &   &        & 0

\end{pmatrix}.
\end{equation*}
\[  Z(M^*) = r\left( (m-1) + (m-3) + \cdots + (m+1 - 2s)\right) + p(m-1-2s) = r(m-s)s + pm - p - 2sp =\] 
\[ = m(sr + p) - rs^2 - p - 2sp = m^2 - \dfrac{(rs)^2 + 2rsp + pr}{r}   \leq m^2 - \dfrac{(rs)^2 + 2rsp + p^2}{r}    = m^2\left(1 - \dfrac1{r}\right) .        \]
Then 
\[ S(M) \leq S(M^*) = \dfrac{Z(M^*)}{m^2} \leq \left(1 - \dfrac1{r}\right). \]

\end{proof}

\end{lem}

\begin{proof}[Proof of Theorem \ref{TB3}]
The first $s$ columns of $M$ form the matrix $M' \in \mathbb{R}^{n\times s}$ with  $\mathrm{rank}(M') = r' \leq r$. Using Lemma \ref{premain}, we conclude that there exists $\tilde{M'} \in \mathbb{R}^{r'\times s}$ such that $S(M')\leq S(\tilde{M'})$. Applying Lemma \ref{main} we get $S(M')\leq S(\tilde{M'}) \leq \left(1 - \dfrac1{r'}\right) \leq \left(1 - \dfrac1{r}\right) $. Then from \eqref{ko}:
 
\begin{equation*}
\dfrac{\sum\limits_{i, j =1}^{s,s}F(M_i,M_j)^2} {s^2} \geq \dfrac1{r^2}.
\end{equation*}
Then from \eqref{first} for every distribution $q$ we obtain: 
\begin{equation}
\label{lower_quad} \sum\limits_{i,j}^{m}q_iq_jF(M_i,M_j)^2 \geq \dfrac{1}{(\ln m +1)^2r^2}.
\end{equation}
And finally, using \eqref{B3def},
\[ B_3(P) \leq (\ln m +1)^2r^2 .\]

\end{proof}

\subsection{Upper bound on $B_4$}

\begin{thm}\label{B4-bnd}
Let $M \in \RR^{n\times m}$ be a stochastic matrix, $\rank M = r$. Then

\begin{equation}
\label{B4bound}  B_4(M) \leq r. 
\end{equation}

\end{thm}
\begin{proof}
Again, we apply the row elimination transformation. Note that since every row in the matrix $M$ after this transformation is either multiplied by some nonnegative factor $\alpha$ or remains unchanged, the maximal element in this row is, obviously, multiplied by the same factor $\alpha$ or remains constant as well.

Suppose $M$ has at least $r+1$ non-zero rows, and without loss of generality, suppose that these are the first $r+1$ rows of M. Now consider the $\e$-transformation $M\ep$ of $M$, and explore how the functional $B_4$ changes after such transformation, taking the last remark into consideration:

\[ B_4(M\ep) - B_4(M) = \sum_{i=1}^n \left( \max_j M\ep(i,j) - \max_j M(i,j) \right)  = \sum_{i=1}^n \left( (1 + \alpha_i\e)\max_j M(i,j) - \max_j M(i,j) \right)  = \]
\[   =    \sum_{i=1}^n \left(\alpha_i\e\max_j M(i,j)\right) = \e\cdot\Lambda.   \]

Similarly to previous proofs, $B_4$ is linear in terms of $\e$, and therefore when $\e$ equals one of the ends of $\Delta_{\alpha}$, the difference between $B_4(M\ep)$ and $B_4(M)$ is nonnegative, while $M\ep$ has strictly less non-zero rows, then $M$. Again, applying such transformations with suitable $\e$'s, at the end we obtain the matrix $\tilde{M}$ with at most $r$ non-zero rows, for which $B_4(M) \leq B_4(\tilde{M})$. It only remains to note that in the formula for $B_4(\tilde{M})$ there are at most $r$ non-zero summands, each less or equal than $1$ (since $\tilde{M}$ is also stochastic). Therefore, we have $B_4(M) \leq B_4(\tilde{M}) \leq r$.

\end{proof}

\subsection{Upper bound on $B_5$}
\begin{thm}
Let $M \in \RR^{n\times m}$ be a stochastic matrix, $\rank M = r$. Then

\[  B_5(M) \leq r^2(\ln m +1).  \]
\end{thm}
\begin{proof}
Simply applying \eqref{lower_quad} and \eqref{B4bound} 
, we get:
\begin{equation*}
B_5(M) \leq \sum_{i=1}^n \max_{\{q^{(i)}_j\}_{j=1}^m} \left( (\ln m + 1)r \sum_{k=1}^mq_k^{(i)}M(i, k) \right) = (\ln m + 1)r\sum_{i=1}^n \max_k M(i,k) \leq (\ln m +1)r^2.
\end{equation*}
The last inequality is due to Theorem~\ref{B4-bnd}.
\end{proof}

\bibliographystyle{alpha}
\bibliography{psd_bounds}
\nocite{*}

\end{document}